\documentclass{llncs}
\usepackage{graphicx, subfig, amsfonts, pslatex}
\usepackage{hyperref}
\usepackage[lined,ruled,commentsnumbered]{algorithm2e}

\hypersetup{colorlinks=false, pdfborder={0 0 0}}

\begin{document}

\title{Approximating the Edge Length of $2$-Edge Connected Planar
Geometric Graphs on a Set of Points\thanks{This is the extended version
of a paper with the same title that will appear in the proceedings of the
10th Latin American Theoretical Informatics Symposium (LATIN 2012), April 16-20,
2012, Arequipa, Peru.}\\(Extended Version)}
\author{Stefan Dobrev\inst{1} \and Evangelos Kranakis\inst{2} \and  
Danny Krizanc\inst{3} \and 
Oscar Morales-Ponce\inst{4} \and Ladislav Stacho\inst{5}}
\institute{
Institute of Mathematics, Slovak Academy of Sciences, 
Bratislava, Slovak Republic. Supported in part by VEGA and APVV grants.
\and
School of Computer Science, Carleton University,
Ottawa, ON,  K1S 5B6, Canada. Supported in part by NSERC and MITACS grants.
\and
Department of Mathematics and Computer Science, Wesleyan University,
Middletown CT 06459, USA. 
\and
School of Computer Science, Carleton University,
Ottawa, ON,  K1S 5B6, Canada. Supported by MITACS Postdoctoral Fellowship.
\and
Department of Mathematics, Simon Fraser University, 8888 University
Drive, Burnaby, British Columbia, Canada, V5A 1S6. Supported in part by NSERC grant.
}



\maketitle

\begin{abstract}
Given a set $P$ of $n$ points in the plane, we solve
the problems of constructing a  
geometric planar graph spanning $P$
1) of minimum degree 2, and 2) which is 2-edge connected, respectively, 
and has max edge length bounded by a factor of 2 times the optimal; 
we also show that the factor 2 is best possible given
appropriate connectivity conditions on the set $P$, respectively. 
First, we construct in $O(n\log{n})$ 
time a geometric planar graph of minimum degree 2 and 
max edge length bounded by 2 times the optimal. 
This is then used to construct in $O(n\log n)$ time a 2-edge 
connected geometric planar graph 
spanning $P$ with max edge length bounded by $\sqrt{5}$ times 
the optimal, assuming that the set $P$ forms a connected
Unit Disk Graph. Second, we prove that 2 times 
the optimal is always sufficient if the set of
points forms a 2 edge connected Unit Disk Graph 
and give an algorithm that runs 
in $O(n^2)$ time. We also show that 
for $k \in O(\sqrt{n})$, there exists a set $P$ of $n$ points in 
the plane such that even though the Unit Disk 
Graph spanning $P$ is $k$-vertex connected, there is no 
2-edge connected geometric planar graph 
spanning $P$ even if the length of its edges is allowed 
to be up to 17/16.
\end{abstract}

\section{Introduction}

Consider a set of points $P$ in the plane in general position, 
and a real number $r\ge0$, the radius. The geometric graph $U(P,r)$ 
is the graph spanning $P$ in which two vertices are joined 
by a straight line iff their (Euclidean) distance is at most $r$. 
Note that the geometric graph $U(P,1)$ is the well known 
unit disk graph on $P$, and in fact $U(P,r)$ is a unit disk 
graph for any $r$ when $r$ is considered to be the unit.

The main focus of this paper is to find 2-edge connected geometric free crossing
(or planar) graphs on a set of points such that the longest edge is minimum.  
Recall that a graph $G$ is 2-edge connected if the removal of any edge does not 
disconnect $G$. 
Several routing algorithms have been designed for 
planar subgraphs of Unit Disk Graphs, for example~\cite{urrutia:2007:LSG}, which are widely 
accepted as models for wireless ad-hoc networks.
Therefore it would be essential for the robustness of 
routing algorithms to construct such 
geometric graphs with ``stronger'' connectivity characteristics. 

Observe that the optimal length of any 2-edge connected geometric
planar graph on a set of points $P$ is at least the
min radius to construct a 2-edge connected UDG  on $P$ possible with crosses.  
Thus, we can raphase the problem as follows:  For what connectivity assumptions on 
$U(P,1)$ and for what $r$ does the geometric graph $U(P,r)$ have a 2-edge connected geometric
planar subgraph spanning $P$? Clearly, $r$ gives an approximation
to the optimal range  when the connectivity of $U(P,1)$ is at most 2-edge connected.

\subsection{Related work}

Two well-known constructions are related to this problem. 
If $U(P,1)$ is connected, then the well-known Gabriel Test 
(see~\cite{Gabriel-69}~and~\cite{toussaint-80}) will 
result in a planar subgraph of $U(P,1)$. However,
2-edge connectivity is not guaranteed. 
Alternatively, the well-known Delaunay Triangulation 
on $P$ 
will result in a 2-edge connected planar subgraph of $U(P,r)$.  However the 
radius $r$ (the length of the longest edge of this triangulation) is not
necessarily bounded.

 Abellanas et al. \cite{abellanas-08} give a polynomial algorithm 
 which augments any geometric planar graph 
to 2-vertex connected or 2-edge connected geometric planar graph, 
respectively, but no bounds are given on the
length of the augmented edges. T\'oth \cite{toth-08} improves the 
bound on the number of necessary edges in such 
augmentations, and Rutter and Wolff \cite{rutter-08} prove 
that it is NP-hard to determine the minimum number 
of edges that have to be added in such augmentations. 

T\'oth and Valter \cite{toth-09} characterize geometric planar 
graphs that can be augmented to 3-edge 
connected planar graphs. Later Al-Jubeh et al. \cite{al-jubeh-09} 
gave a tight upper bound on the number of 
added edges in such augmentations. Finally, Garc\'ia et al.
 \cite{garcia2009triconnected} show how to construct 
a 3-connected geometric planar graph on a set of points in 
the planar with the minimum number of straight line edges of unbounded length.

A related problem is studied in \cite{morales-5}. The 
authors prove that it is NP-hard to decide whether 
$U(P,\frac{\sqrt5}2)$ contains a spanning planar graph of 
minimum degree 2 even if $U(P,1)$ itself has minimum degree 2. 
They also posed and studied the problem of finding the minimum radius $r$ so 
that $U(P,r)$ has a geometric planar spanning subgraph of 
minimum degree 3 provided that $U(P,1)$ has a spanning 
subgraph of minimum degree 3.

Closely related is the research by Kranakis et al.~\cite{kranakis-10} which
shows that if $U(P,1)$ is connected then $U(P, 3)$ has a
2-edge connected geometric planar spanning subgraph. The 
construction starts from a minimum 
spanning tree of $U(P,1)$ which in turn is augmented to a 
2-edge connected geometric planar spanning subgraph of $U(P,3)$.
In the same paper several other constructions are given
(starting from more general connected planar subgraphs) and also 
bounds are given on the minimum number of augmented
edges required.
However, the question of providing an algorithm
for constructing the smallest $r>0$ such that $U(P, r)$ has a
2-edge connected geometric planar spanning subgraph
remained open. This question turns out to be the main focus of our
current study. 

Our problem is also related to the well-known bottleneck traveling 
salesman problem, i.e.  finding a Hamiltonian cycle that 
minimizes the length of the longest edge, since such
a cycle is 2 edge conected (but not necessarily planar). Parker et al. \cite{parker} 
gave a 2-approximation algorithm for this problem 
and also showed that there is no better algorithm unless $P=NP$.
There is also literature on constructing $2$ edge connected 
subgraphs with minimum number of edges.
In \cite{cheriyan1998} it is
proved that given a $2$-edge connected graph there is an algorithm
running in time $O(m n)$ which finds a $2$-edge connected
spanning subgraph whose number of edges is $17/12$ times the optimal,
where $m$ is the number of edges and $n$ the number of vertices
of the graph. An improvement  is provided in \cite{Vempala-00} 
in which a 4/3 approximation algorithm is given. 
Later, Jothi et al. \cite{jothi-03} provided 
a 5/4-approximation algorithm. 
However in these results the resulting
spanning subgraphs are not guaranteed to be planar.

\subsection{Contributions and outline of the paper}

We start with Section~\ref{sec:prel}, where we give 
the notation and provide some concepts which are useful for the proofs. 
In Section~\ref{sec:mindegree} we prove that if $U(P,1)$ has 
minimum degree 2, then $U(P, 2)$ contains a 
spanning geometric planar subgraph with minimum degree 2. 
Note that these subgraphs are not necessarily connected. 
An algorithm that runs in time $O(n\log{n})$ to find 
such a subgraph is presented as well.  
In Section~\ref{sec:2edge} we prove that if $U(P,1)$ is 
connected and has minimum degree 2, 
then $U(P, \sqrt{5})$ contains a 2-edge connected  
spanning  geometric planar subgraph and we give a corresponding 
algorithm that runs in time $O(n\log{n})$.  
In section~\ref{sec:2edge2} we combine results from previous sections and
prove the main theorem of the paper by showing that if $U(P,1)$ is 2-edge connected, 
then  $U(P, 2)$ contains a 2-edge connected spanning geometric planar subgraph. 
A corresponding algorithm that runs in 
time $O(n^2)$ is presented as well. We also show that all the bounds 
are tight. In Section~\ref{sec:highconnected} we 
show that there exists a set $P$ of $n$ points in the plane so that 
$U(P,1)$ is $k$-vertex connected,  $k \in O(\sqrt{n})$,  
but even $U(P,17/16)$ does not contain any 2-edge connected spanning geometric  planar subgraph.

\section{Preliminaries and Notation}
\mbox{\label{sec:prel}}

Let $G=(V, E)$ be a connected graph. As usual we represent 
an undirected edge as $\{u,v\}$ and a directed edge 
with head $u$ and tail $v$ as $(u,v)$. A vertex $v \in V$ 
is a cut-vertex of $G$ if its removal disconnects $G$.
Similarly an edge $\{u,v\} \in E$ is a cut-edge or bridge 
if its removal disconnects $G$.
We denote the line segment between two points $x$ and
$y$ by $xy$ and their (Euclidean) distance by 
$d(x,y)$. Let $C(x;r)$ denote the circle of 
radius $r$ centered at $x$, and let $D(x;r)$ denote the 
disk of radius $r$ centered at $x$.

Before we proceed with the main results of the paper we introduce 
the concepts of $Tie$ and $Bow$ that will
help to distinguish various crossings in the 
proof of the main results. 
\begin{definition}
\label{def:pico}
We say that four points $u, v, x, y$ form a {\em Tie}, denoted by 
$Tie(u; v, x, y)$, if $uv$ crosses $xy$, $x$ and
 $y$ are outside of $D(u; d(u, v))$ and $u$ is outside 
 of $D(x; d(x,y))$. The point $u$ is called the {\em tip} of the
 $Tie$ and $xy$ the crossing line of $\{u, v\}$. 
 See Figure~\ref{fig:tie}.
\end{definition}

\begin{lemma}
\label{lem:pico}
  Let $u,v,x,y$ form a $Tie(u; v,x,y)$. Then, $\pi/3 \leq\angle(uvx) < 2\pi/3$ and $\pi/3 \leq \angle(yvu) < 2\pi/3$.
\end{lemma}

\begin{proof}
Consider the angle $\angle (yvx)$. 
Observe that $\angle (yvx) \geq \pi/2$ since by Definition~\ref{def:pico}, 
$x,y \notin D(x; d(x,y))$ and $uv$ crosses $xy$.  
Therefore, $d(x, y) > \max(d(x, v), d(v, y))$. Also from Definition~\ref{def:pico}, $d(u, x) > d(x, y)$. 
 Therefore, $\angle(uvx) \geq \pi/3$ since it is the largest angle in the 
triangle $\triangle(uvx)$. It remains to prove that $\angle(yvu) \geq \pi/3$ 
and the result follows since $\angle(yvx) < \pi$.
 For the sake of contradiction assume that $\angle(yvu) < \pi/3$; 
 see Figure~\ref{fig:tie2}.  From
 Definition~\ref{def:pico}, $d(u, v) < d(u, y)$.  
 Hence, $\angle(uyv) < \angle(vuy)$ and consequently $\angle(vuy)$ 
is the largest angle in $\triangle(vuy)$. 
Therefore, $\angle(xuy) > \angle(vuy) > \angle(uyv) > \angle(uyx)$ which implies that 
$d(x, y) > d(u,x)$.  This contradicts Definition~\ref{def:pico}.

\begin{figure}[!ht]
  \centering
  \includegraphics[scale=.7]{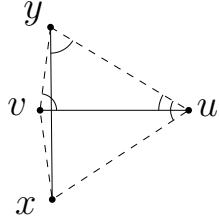}
  \caption{If $u,v,x,y$ form a $Tie(u;v,x,y)$, then $\angle(yvx)
    \geq 2\pi/3$.}
  \label{fig:tie2}
\end{figure}
\qed
\end{proof}

\begin{lemma}
\label{lem:pico2}
  Let $u,v,x,y$ form a $Tie(u; v,x,y)$ and $u'$ be a point.\\
  (i)  If $u'v$ crosses $ux$, then $u',v,u,x$ cannot form a $Tie$.\\ 
  (ii) If $u'x$ crosses $uv$, then $u',x,u,v$ cannot form a $Tie$.
\end{lemma}

\begin{proof}
(i) Arguing by contradiction, assume that $u'v$ and $ux$ form a 
$Tie(u'; v, u, x)$; see Figure~\ref{fig:uniquetie}. 
From Lemma~\ref{lem:pico}, $\angle(xvu) \geq 2\pi/3$.  
Now consider the $Tie(u; v,x,y)$. From Lemma~\ref{lem:pico},
 $\angle(uvx) < 2\pi/3$, a contradiction.

(ii) From Lemma~\ref{lem:pico}, $\angle(uvx) \geq \pi/3$. 
Therefore, $\angle(vxu) < 2\pi/3$. However, 
the minimum angle $\angle(uxv)$ to form a $Tie(u'; x, u, v)$ is at 
least $2\pi/3$; see Figure~\ref{fig:uniquetie1}.

\begin{figure}[!ht]
  \centering \subfloat[$\{u',v\}$ and $\{u,x\}$ cannot form a
  $Tie$.]{\label{fig:uniquetie}\includegraphics[scale=.7]{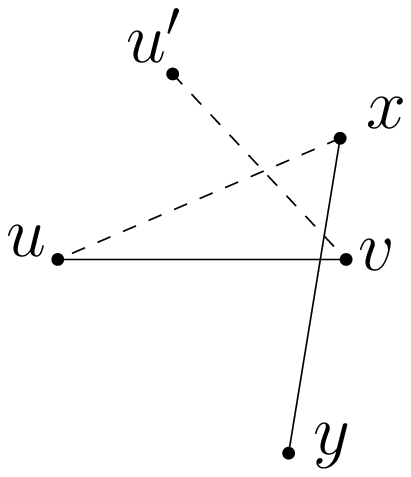}}
  \hspace{2cm}%
  \subfloat[$\{u',x\}$ and $\{u,v\}$ cannot form a
  $Tie$.]{\label{fig:uniquetie1}\includegraphics[scale=.7]{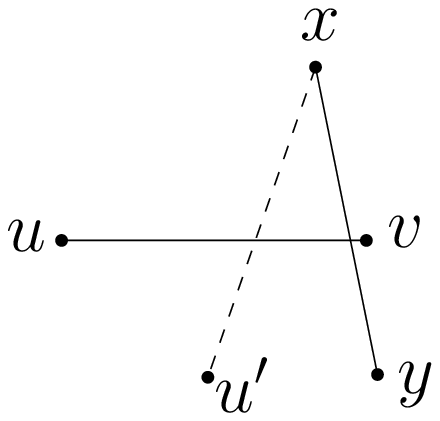}}
  \caption{If $u,v,x,y$ form a $Tie(u; v,x,y)$, then $u'$ cannot
    form a $Tie$ with either $v$ or $x$ or $y$ that overlaps
    $Tie(u; v,x,y)$. }
\end{figure}
\qed
\end{proof}

The following lemma shows that the points of a $Tie(u;v,x,y)$ are at distance at most $\sqrt{2}$ of each other.

\begin{lemma}
\label{lem:pico3}
  Let $u,v,x$, and $y$ be four points forming a $Tie(u;v,x,y)$ such that $\max\{d(u, v),$ $d(x,y)\} =1$. Then, $d(u, x)$ and $d(u, y)$ are bounded by $\sqrt{2}$.
\end{lemma}

\begin{proof}
  Let $p$ be the intersection point of $xy$ and $C(u; d(u,v))$ closer to $y$, and $l$ be the tangent line at $p$;
 see Figure~\ref{fig:boundedlength}.  Since the angle that $up$ forms with $l$ is $\pi/2$, $\angle(upx) \leq \pi/2$. 
Therefore, $d(u, x) \leq \sqrt{2}$, since $\max(d(u, p), d(p, x)) \leq 1$. Similarly, we can prove that $d(u,y) \leq \sqrt{2}$.

\begin{figure}[!ht]
  \centering
  \includegraphics[scale=.7]{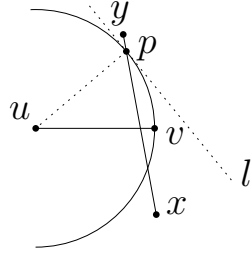}
  \caption{$d(u, x) \leq \sqrt{2}$ and $d(u, y) \leq \sqrt{2}$ in a $Tie(u;v,x,y)$.}
  \label{fig:boundedlength}
\end{figure}
\qed
\end{proof}

We conclude the preliminaries by introducing the concept of a $Bow$.

\begin{definition}
\label{def:bucket}
We say that four points $u, v, x, y$ form a {\em Bow}, denoted by $Bow(u, v, x, y)$, if $uv$ crosses $xy$, $d(u, y) \leq d(u,v) < d(u,x)$ and $d(v, x) \leq d(x,y) < d(u,x)$. 
See Figure~\ref{fig:bow}.
\end{definition}

\begin{figure}[!ht]
  \centering
\subfloat[$Tie(u;v,x,y)$ with tip $u$.]{\label{fig:tie}\includegraphics[scale=.55]{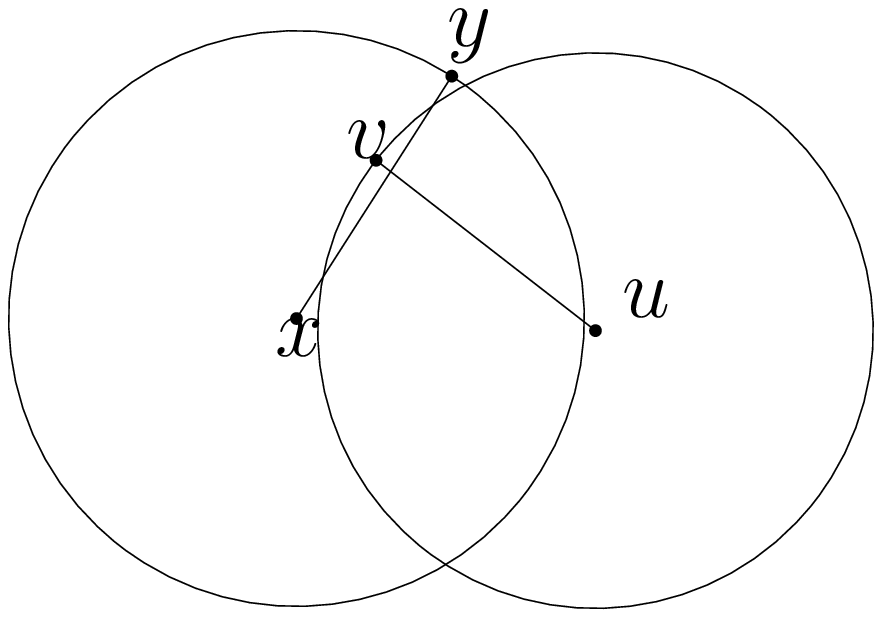}}
\hspace{2cm}%
\subfloat[$Bow(u,v,x,y)$]{\label{fig:bow}\includegraphics[scale=.55]{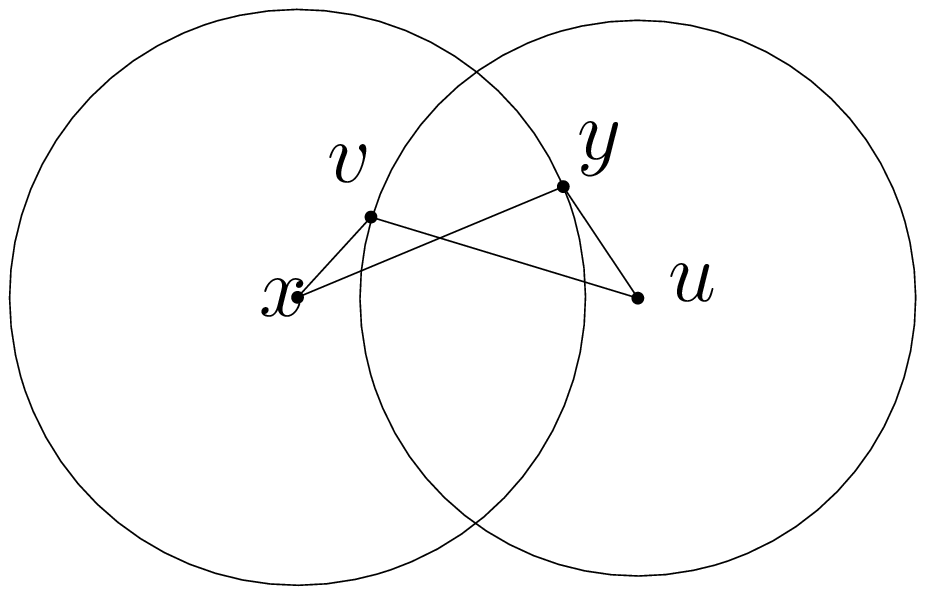}}
  \caption{Tie and Bow.}
\end{figure}

\section{Planar Subgraphs of Minimum Degree 2 of a UDG of Minimum Degree 2}
\label{sec:mindegree}

In this section we prove that if $U(P,1)$ has minimum degree 2, then $U(P,2)$ always 
contains a spanning geometric planar subgraph of minimum degree 2. We also show that the radius $2$ is best possible. 
Therefore in this section we assume $U(P,1)$ has minimum degree 2.


The following theorem shows that the bound $2$ is the best possible. 

\begin{theorem}
\label{thm:lowerthm1}
For any real $\epsilon>0$ and any integer $k$, there exists a set $P$ of $4k$ points in the plane so that $U(P,1)$ has minimum degree 
2 but $U(P,2-\epsilon)$ has no geometric planar spanning subgraph of minimum degree 2.
\end{theorem}

\begin{proof}
  It is not difficult to see that the component depicted in Figure~\ref{fig:dmin2} requires $\{u,v\}$ to create a planar graph of 
degree two. To create a family of UDGs with $4k$ vertices, it is enough to consider $k$ disconnected components.

  \begin{figure}[!ht]
    \begin{center}
      \includegraphics[scale=.7]{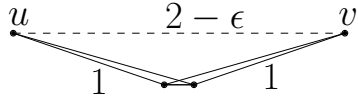}
    \end{center}
    \caption{UDG of minimum degree two that requires scaling factor of $2-\epsilon$.}
    \label{fig:dmin2}
  \end{figure}
\qed
\end{proof}


Let $T = (P, E)$ be the minimum spanning forest (MSF) (or nearest neighborhood graph) of $U(P,1)$ formed by connecting each vertex with its 
 neareast neighbor. Recall that $U(P,1)$ has minimum degree 2 but it is not guaranteed to be connected, and  that any 
two vertices in different components are at distance more than $1$.  Let $u$ be a leaf of $T$ and $v$ be the second nearest neighbor of $u$. (If there exist more than one, then choose any one among them.) The directed edge $(u,v)$ is defined as a second nearest neighbor edge (SNN edge). Let $E'$ be the set of SNN edges for all leaves of $T$.  Observe that $E \cap E' = \emptyset$, since the nearest neighborhood graph is a subgraph of $U(P,1)$ and SNN edges of $E'$ are considered for leaves of $T$.

Before giving the main theorem we provide some lemmas that are required for the proof.  The following lemma shows that if an SNN edge $(x,y)\in E'$ crosses an edge $\{u,v\}$ of $T$, then the four vertices form a $Tie(u; v, x, y)$.

\begin{lemma}
\label{lem:crossmst} 
Let $(x,y)\in E'$ be an SNN edge that crosses an edge $\{u, v\} \in T$. Then, the four vertices form a $Tie(u; v, x, y)$ such that either $\{u, x\} \in T$ or $\{v, x\} \in T$.  Moreover, the quadrangle $uxvy$ is empty.
\end{lemma}

\begin{proof}
First we will show that if $(x,y)$ crosses $\{u, v\}$ then either $\{u, x\} \in T$ or $\{v, x\} \in T$.  For the sake of contradiction, assume that neither $\{u,x\} \notin T$ nor $\{v,x\} \notin T$.  Observe that $u$ and $v$ are outside $D(x; d(x,y))$, otherwise $(x,y)$ would not be the SNN edge; see Figure~\ref{fig:crossmst}.  Therefore, $\angle(vyu) \geq \pi/2$ since $(x,y)$ crosses $\{u,v\}$.  Hence, $d(u,v)$ is greater than $d(u,y)$ and $d(v,y)$.  This contradicts the minimality of MSF  $T$, since replacing $\{u,v\}$ by either $\{u,y\}$ or $\{v,y\}$ results in a spanning forest of $U(P,1)$ of smaller weight.

To show that the four vertices form a $Tie(u; v, x, y$), assume that $\{v, x\} \in T$. Observe that $d(u, x) > d(x, y) > \max\{d(v, x), d(v, y)\}$ since $y$ is the second nearest neighbor of $x$ and $\angle(xvy) \geq \pi/2$; see Figure~\ref{fig:snntwocross}. It is not difficult to see that $d(u, v) < d(u, x)$ (Otherwise we can obtain a spanning forest of smaller weight by replacing $\{u,v\}$ with $\{u, x\}$.) To prove that $d(u,v) < d(u, y)$ assume by contradiction that $d(u,v) > d(u, y)$. 
Hence, $\angle(yuv)$ is the largest angle in $\triangle(uvy)$ since $d(u, v) < d(v,y)$ (Otherwise we can obtain a spanning forest of smaller weight by replacing $\{u,v\}$ with either $\{u, y\}$ or $\{v, y\}$.) Therefore, $\angle(yux) > \angle(yuv)$ which implies that $d(x, y) > d(u,x)$. This is a contradiction since $d(x, y) < d(u,x)$.

To prove that $uxvy$ is empty, we consider independently $\triangle(uvx)$ and $\triangle(uvy)$.  First consider $\triangle(uvx)$.  It is known that the angle that a vertex forms with two consecutive neighbors in $T$ is at least $\pi/3$ and the triangle is empty.  Therefore, $v$ does not have a neighbor in the sector $\angle(xvu)$ since by Lemma~\ref{lem:pico} $\angle(uvx) < 2\pi/3$. Therefore, $\triangle(uvx)$ is empty.  Now we consider $\triangle(uvy)$.  Assume by contradiction that exists a point $p$ in $\triangle(uvy)$ as depicted in Figure~\ref{fig:crossmst3}.  Observe that $\angle(uvp) > \pi/3$ (Otherwise we can replace $\{u, v\}$ with either $\{u, p\}$ or $\{v,p\}$.) Therefore, $\angle(xvp) < \angle(xvy)$ and $d(x, p) < d(x,y)$ since $d(v, p) \leq d(v, y)$ which contradicts the SNN edge definition.

\begin{figure}[!ht]
  \centering \subfloat[$\{u, v\} \notin
  D(x;d(x,y))$]{\label{fig:crossmst}\includegraphics[scale=.65]{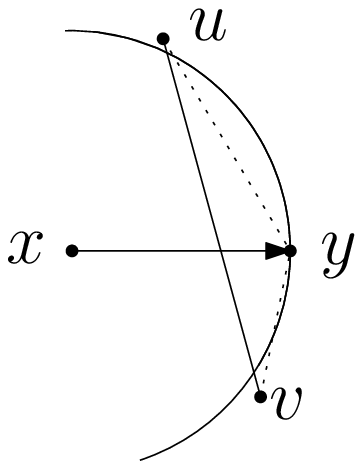}}
  \hspace{1cm}%
  \subfloat[A SNN edge that crosses an edge of $T$ forms a
  $Tie$.]{\label{fig:snntwocross}\includegraphics[scale=.65]{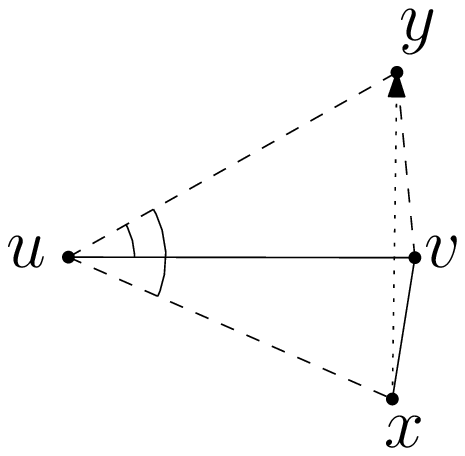}}
  \hspace{1cm}%
  \subfloat[$uxvy$ is
  empty.]{\label{fig:crossmst3}\includegraphics[scale=.65]{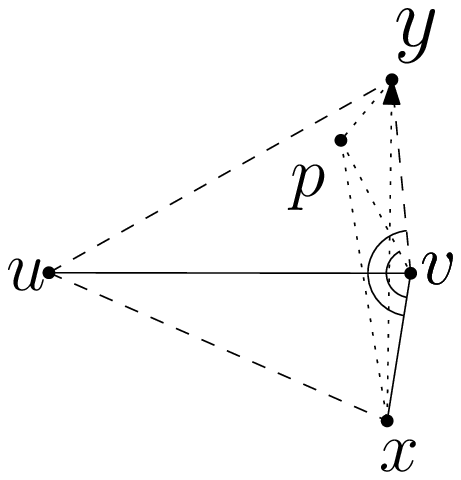}}
  \caption{A SNN edge crossing an edge of $T$}
\end{figure}
\qed
\end{proof}

As a consequence of Lemma~\ref{lem:crossmst}, an SNN edge crosses at most one edge of $T$, since the angle that a vertex forms with two consecutive neighbors in $T$ is at least $\pi/3$.  The following lemma will help to characterize crossings between SNN edges.

\begin{lemma}
\label{lem:crossvirtual}
Let $(u,v), (u',v')\in E'$ be two crossing SNN edges. Then $\{u', v\}
  \in T$.
\end{lemma}

\begin{proof}
Assume that $\{u', v\},\{u,v'\} \notin T$, then $u'$ and $v'$ are not in $D(u; d(u,v))$ as depicted in Figure~\ref{fig:crosssnn}. Observe that if either $u'$ or $v'$ is in $D(u; d(u,v))$, then $(u,v)$ would not be the SNN edge. Therefore, $d(u',v') > \max(d(u',v), d(v,v'))$ since $\angle(v'vu') > \pi/2$ and $(u,v)$ crosses $(u',v')$. This is a contradiction since $\{u',v\} \notin T$.

\begin{figure}[!ht]
  \centering
  \includegraphics[scale=.7]{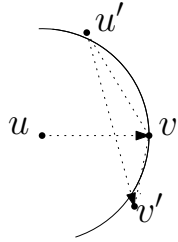}
  \caption{Two crossing SNN edges}
  \label{fig:crosssnn}
\end{figure}
\qed
\end{proof}

\begin{lemma}
\label{lem:crossvirtual2}
Let $(u,v), (u',v')\in E'$ be two crossing SNN edges. \\
(i) If $\{u, v'\},\{u',v\} \in T$, then they form a $Bow(u, v,
  u', v')$ such that the quadrangle $uv'vu'$ is empty. \\
(ii) If $\{u', v\} \in T$ and $\{u, v'\} \notin T$, then they form a
  $Tie(u; v, u', v')$ such that the quadrangle $uu'vv'$ is either
  empty or contains the neighbor of $u$ in $T$.
\end{lemma}

\begin{proof}
(i) Let $\{u, v'\} \in T$ and $\{u',v\} \in T$.  Clearly, $d(u, u') > d(u, v) > d(u, v')$, since $v'$ is the nearest neighbor of $u$ and $v$ the second.  Similarly, $d(u, u') > d(u', v') > d(u', v)$.  Therefore, the four vertices form a $Bow(u, v, u', v')$. To prove that the quadrangle $uv'vu'$ is empty consider $R = D(u; d(u, v)) \cup D(u'; d(u',v'))$ as depicted in Figure~\ref{fig:crossbow}.  Obviously any point inside $R$ is closer to either $u$ or $u'$. Therefore, $R$ contains only $u,v,u',v'$.

(ii) Let $\{u',v\} \in T$ and $\{u, v'\} \notin T$. From the definition of SNN edge, $d(u,v) \leq \min\{d(u,  u'), d(u, v')\}$ and $d(u',v')<d(u, u')$. Therefore, the four vertices form a $Tie(u; v, u', v')$. To prove that the quadrangle may contain at most one point $p$ such that $\{u, p\} \in T$, consider $R = D(u; d(u, v)) \cup D(u'; d(u',v'))$ as depicted in Figure~\ref{fig:crosstie}.  Obviously any point inside $R$ is closer to either $u$ or $u'$. Therefore, it contains only the nearest neighbors of $u$ and $u'$. Further, $v$ is the nearest neighbor of $u'$. Therefore, $p \in R$ where $\{u, p\} \in T$.  It remains to prove that $R$ contains the quadrangle $uu'vv'$. Let $a$ be the intersection point of $\{u,v'\}$ and $C(u; d(u, v))$. It is enough to prove that $a \in D(u'; d(u', v'))$.  However, $\angle(u'va) < \angle(u'vv') $ and $\angle(avv') < \pi/3$.  Therefore, $d(u', a) < d(u', v')$.

  \begin{figure}[!ht]
    \centering \subfloat[$\{\{u, v'\},\{u',v\}\} \in
    T$]{\label{fig:crossbow}\includegraphics[scale=.7]{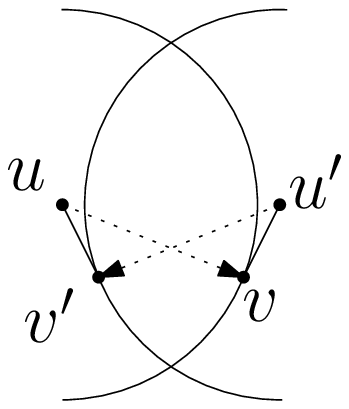}}
    \hspace{2cm}%
    \subfloat[$\{u', v\} \in T$ and $\{u, v'\} \notin
    T$]{\label{fig:crosstie}\includegraphics[scale=.7]{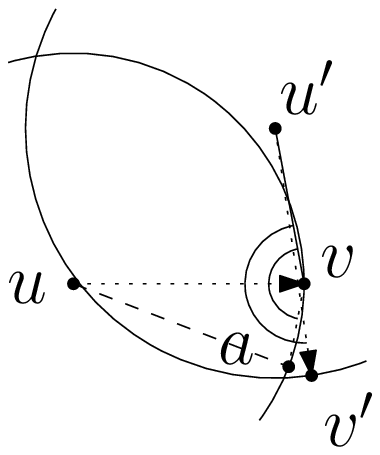}}
    \caption{Crossings of SNN edges}
  \end{figure}
  \qed
\end{proof}

The following lemma will help to determine our upper bound.

\begin{lemma}
\label{lem:length}
Let $u,v,u',v'$ be four vertices forming a $Tie(u; v, u', v')$ and $w$ be a vertex such that $d(u, w) \leq 1$, $\angle(wuv) \leq \varphi$, and $\{u', u\}$ crosses $\{w, v\}$.  Then, $d(w, u')^2 \leq 3-2\sqrt{2}\cos(\varphi - \pi/4)$.
\end{lemma}

\begin{proof}
Observe that $\{u',v'\}$ crosses at least two points of $C(u; d(u,v))$.  Thus, we can assume without loss of generality that $\{u',v'\}$ crosses $C(u; d(u,v))$ in $v$ and $d(u, v) = d(u, v')$ as depicted in Figure~\ref{fig:tielength}.  Let $\alpha = \angle(vuv')$ and $\beta = \angle(uv'v) = \angle(v'vu) = \frac{\pi - \alpha}{2}$.  Observe that $0 < \alpha \leq \pi/3$ since by Lemma~\ref{lem:pico}, $\angle(uvv') \geq \pi/3$.  By the law of cosines in $\triangle(uv'u')$, $d(u, u')^2 = d(u, v')^2 + d(u',v')^2 - 2d(u, v')d(u',v')\cos(\beta) \leq 2 -2\cos(\beta)=2-2\sin(\alpha/2)$ and $d(u, u') \leq 2\sin(\frac{\beta}{2})=2\cos(\frac{\pi - \alpha}{4})$.  

Let $\gamma = \angle(wuu') = \varphi - \angle(u'uv)$.  Since $\angle(v'vu) = \beta$, $\angle(uvu') = \pi - \beta$.  Therefore, if $d(u, v) \leq d(u', v)$, then $\angle(u'uv) \geq \frac{\pi - (\pi - \beta)}{2} = \frac{\pi - \alpha}{4}$. Otherwise, $\angle(vu'u) \geq \frac{\pi - (\pi - \beta)}{2} = \frac{\beta}{2}$.  From $\triangle(uv'u')$, $\angle(u'uv) \geq \pi - \beta - \frac{\beta}{2} - \alpha = \frac{\pi-\alpha}{4}$.

From the law of cosines, $d(w, u')^2 = d(u, w)^2 + d(u, u')^2 - 2d(u,u')d(u,w)\cos(\gamma) \leq 3 -2\sin(\frac{\alpha}{2}) - 4\cos(\frac{\pi - \alpha}{4})\cos(\varphi - \frac{\pi - \alpha}{4})$.  Observe that when the angles satisfy $0\leq \alpha \leq \pi/3$ and $\pi/3 \leq \varphi \leq \pi$, then the three values $\sin(\frac{\alpha}{2}), \cos(\frac{\pi - \alpha}{4})$ and $\cos(\varphi - \frac{\pi - \alpha}{4})$ attain positive values.  Therefore, for any $\varphi \in [\pi/3, \pi]$ the maximum value is reached when $\alpha = 0$ and $d(w, u')^2 \leq 3 - 2\sqrt{2}\cos(\varphi - \frac{\pi}{4})$.

\begin{figure}[!ht]
  \centering
  \includegraphics[scale=.6]{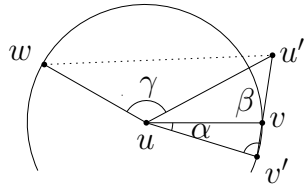}
  \caption{ If $\angle(wuv) \leq \varphi$, then $d(w, u')^2 < 3 -
    2\sqrt{2}\cos(\varphi - \frac{\pi}{4})$}
  \label{fig:tielength}
\end{figure}
\qed
\end{proof}

Now we are ready to prove the main theorem.

\begin{theorem}
\label{thm:main}
Let $P$ be a set of $n$ points in the plane in general position. If  $U(P,1)$ contains a spanning subgraph of minimum degree 2, then $U(P, 2)$ contains a geometric
 planar spanning subgraph of minimum degree 2.  Further, such a subgraph can be constructed in time $O(n\log n)$.
\end{theorem}

\begin{proof}
Consider the Nearest Neighbor Graph $T=(P,E)$ of $U(P,1)$. It is known that $T$ is a subgraph of any minimum spanning tree of $U(P,1)$. Let $E'$ be the set of SNN edges from leaves of $T$.  Clearly every edge in $E'$ has length at most $1$ since $U(P, 1)$ has minimum degree two. Let $G = (P, E \cup E')$. It follows that $G$ spans $P$, has minimum degree 2, however it may not be planar. We show how to modify $G$ to a planar graph. 
	
\begin{claim}
Let $Tie(u; v, u', v')$ be a $Tie$ of $G$ where $u'$ is a leaf of $T$.\\	
(i) $\{u, v\}$ may cross at most one other edge $\{u'', v''\}$
  of $G$ such that they form either a $Tie(v; u, u'', v'')$ or a
  $Tie(u''; v'', u, v)$. \\
(ii) $\{u',v\} \in E$ does not cross any edge of $G$.
\end{claim}

\begin{proof}
(i) From Lemma \mbox{\ref{lem:crossmst}} and Lemma \mbox{\ref{lem:crossvirtual}}, $\{u',v\} \in E$.  Therefore, $v$ is not a leaf in $T$. Hence, if $u$ is a leaf of $T$, then from Lemma \mbox{\ref{lem:crossvirtual}}, $\{u, v\}$ may be only the crossing line of a $Tie(u''; v'', u, v)$ as depicted in Figure \mbox{\ref{fig:twotiesnn2}}. On the other hand, $v$ may be the tip of another $Tie(v;u,u'',v'')$ as depicted in Figure \mbox{\ref{fig:twotiemst}}. However, in that case $u$ is not a leaf of $T$.
 
(ii) Assume by contradiction that $\{u',v\}$ crosses a SNN edge $(x, y)\in E'$ where $x$ is a leaf of $T$.  Therefore, from Lemma \mbox{\ref{lem:crossmst}} they form a $Tie(u'; v, x, y)$ where $\{x,v\} \in E$ since $u'$ is a leaf.  Observe that $(x,y)$ also crosses $(u',v')$ otherwise $(u',v')$ would not be the SNN edge. Therefore, from Lemma \mbox{\ref{lem:crossvirtual}} either $\{v, x\} \in E$ or $\{u', y\} \in E$. This is a contradiction since $u'$ and $x$ are leaves of $T$.

\begin{figure}[!ht]
  \centering \subfloat[$Tie(u; v, u', v')$ and $Tie(u''; v'', u,
  v)$.]{\label{fig:twotiesnn2}\includegraphics[scale=.7]{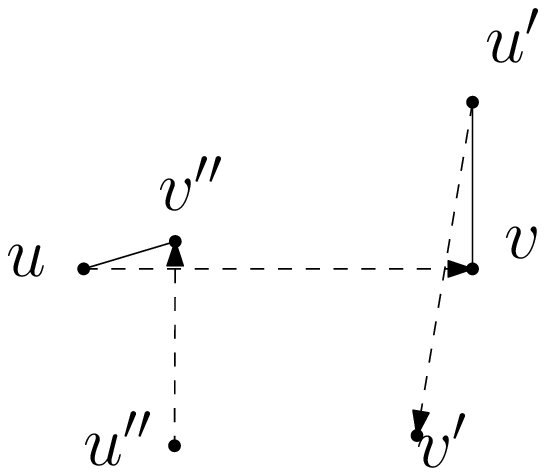}}
  \hspace{2cm}%
  \subfloat[$Tie(u; v, u', v')$ and $Tie(v; u,
  u'',v'')$]{\label{fig:twotiemst}\includegraphics[scale=.7]{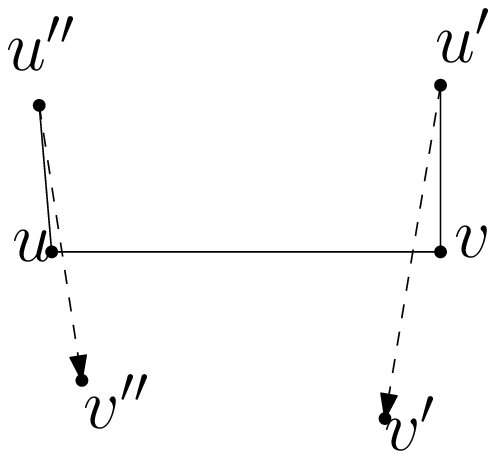}}
  \caption{$\{u,v\}$ is in at most two $Tie$s (Solid lines are
    edges of $T$ and dashed arrow lines are SNN edges.) }
\end{figure}

\end{proof}

The proof is constructive. In every step we remove at least one crossing of $G$ by replacing edges of $E'$. First, we remove all $Tie$s.

Let $Tie(u; v, u', v')$ be a $Tie$ of $G$ where $u'$ is a leaf of $T$.  Observe that from Lemma \mbox{\ref{lem:pico2}}, there is no leaf $r$ of $T$ such that either $(r,v)$ crosses $\{u',v'\}$ or $(r,v')$ crosses $\{u,v\}$.  According to Claim, three cases can occur:

\begin{enumerate}
\item{$\{u,v\}$ does not form another $Tie$.} 
From Lemma~\ref{lem:crossmst} and Lemma~\ref{lem:crossvirtual}, $\triangle(uvu')$ is either empty or it has exactly one vertex $w$ such that $\{w, u\} \in E$.  If $\triangle(uvu')$ is empty, let $E' = E' \cup \{\{u, u'\}\} \setminus \{\{u',v'\}\}$. Otherwise, let $E' = E' \cup \{\{w, u'\}\} \setminus \{\{u',v')\}$; see Figure~\ref{fig:twocross1}.  From Lemma~\ref{lem:pico3}, $d(u, u') \leq \sqrt{2}$. Therefore the length of the new edge is bounded by $\sqrt{2}$. Since $\{u,v\}$ and $\{v, u'\}$ do not cross, the new edge does not cross any edge of $G$.

\begin{figure}[!ht]
  \centering
  \includegraphics[scale=.7]{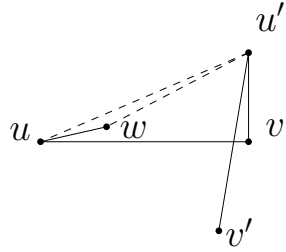}
  \caption{$\{u, v\}$ is in one $Tie$ (Dotted lines are removed
    edges and dashed lines are possible new edges.)}
  \label{fig:twocross1}
\end{figure}

\item{$\{u,v\}$ forms a $Tie(v; u, u'', v'')$ where $u''$ is a leaf of $T$.} 
Observe that in this case $u$ and $v$ are not leaves of $T$.  Therefore, from Lemma~\ref{lem:crossmst} the quadrangles $uu'vv'$ and $vu''uv''$ are empty. We consider two cases. In the first case $\{u, u'\}$ does not cross $\{u'',v\}$. Let, $E' = E' \cup \{ \{u, u' \},\{u'',v\}\} \setminus \{\{u',v'\},\{u'',v''\}\}$ as depicted in Figure~\ref{fig:twocross2}.  From Lemma~\ref{lem:pico3}, the new edges are bounded by $\sqrt{2}$.  In the second case $\{u, u'\}$ crosses $\{u'',v\}$; see Figure~\ref{fig:twocross3}.  Consider the quadrangle $uvu'u''$.  If it is empty, let $E' = E' \cup \{ \{u', u''\} \} \setminus \{\{u',v'\},\{u'',v''\}\}$.  Otherwise, let $p$ and $q$ be the vertices in $uvu'u''$ such that $\angle(uu''p)$ and $\angle(vu'q)$ are minimum.  Let $E' = E' \cup \{\{u', q\},\{u'',q\}\} \setminus \{\{u',v'\},\{u'',v''\}\}$.  From Lemma~\ref{lem:length}, $d(u', u'') \leq 2$ since $\angle(u''uv) \leq 2\pi/3$.  Observe that $p$ does not have a neighbor in the same half-space determined by $\{u'', p\}$ as $u$ because $\angle(uu''p)$ is minimum.  Similarly, $q$ does not have a neighbor in the same half-space determined by $\{u', q\}$ as $v$ because $\angle(vu'q)$ is minimum.  Since, $\{v, u'\}$ and $\{u,u''\}$ do not cross any other edge and $\{u,v\}$ only forms $Tie(u; v, u', v')$ and $Tie(v; u, u'', v'')$, the new edges do not cross any edge of $G$.

\begin{figure}[!ht]
  \centering \subfloat[$\{u, v\}$ is in one
  $Tie$.]{\label{fig:twocross2}\includegraphics[scale=.7]{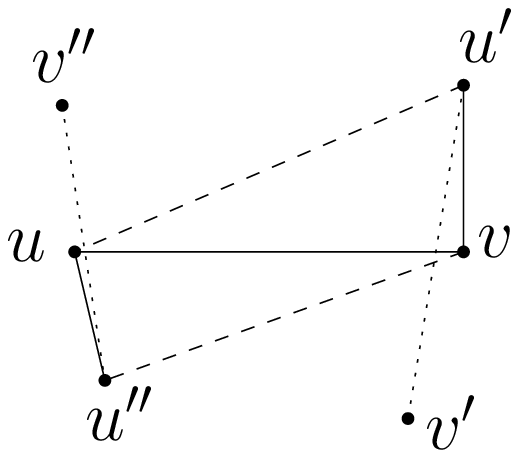}}
  \hspace{2cm}%
  \subfloat[$\{u, v\}$ is in two
  $Tie$s]{\label{fig:twocross3}\includegraphics[scale=.7]{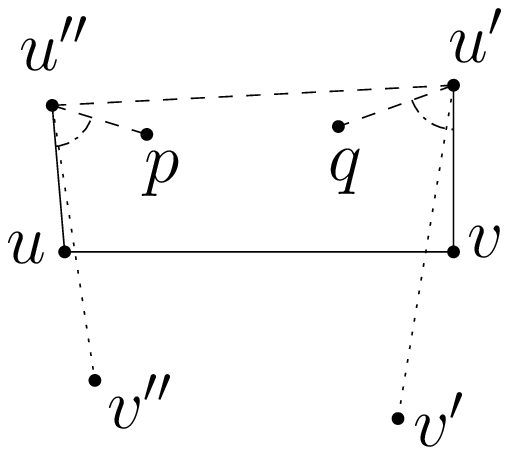}}
  \caption{$\{u, v\}$ crosses at least one edge of $G$ (Dotted lines
    are removed edges and dashed lines are possible new edges.)}
\end{figure}
		 
\item{$\{u,v\}$ forms a $Tie(u''; v'', u,v)$.} 
$\{u,v\}$ forms a $Tie(u''; v'', u,v)$.
Observe that in this case $u$ is a leaf of $T$.  Assume without loss of generality that $\{u'', v\}$ crosses $\{u,u'\}$.  Consider the quadrangle $u''uvu'$. If it is empty, then let $E' = E' \cup \{ \{u', u'' \}\} \setminus \{\{u',v'\}\}$.  Otherwise, let $p$ be the vertex in $u''uvu'$ such that $\angle(vu'p)$ is minimum.  Let $E' = E' \cup \{ \{u', p \}\} \setminus \{\{u',v'\}\}$.  From Lemma~\ref{lem:length}, $d(u', u'') \leq 2$ since $\angle(u''uv) \leq 2\pi/3$.  Observe that all the neighbors of $p$ are in the same half-plane determined by $\{u', p\}$.  It is not difficult to see that the new edge does not cross any edge of $G$ since the region $u''uvu'$ is close.

\end{enumerate}

After removing the $Tie$s we remove the $Bow$s.  Consider a $Bow(u, v, u', v')$ where $u$ and $u'$ are leaves of $T$.  Let $E' = E' \cup \{\{u, u''\}\} \setminus \{\{u,v\}, \{u',v'\}\}$.  Clearly, $d(u, u') \leq 2$ and $\{u, u''\}$ does not cross any edge of $G$.

\LinesNumbered
\begin{algorithm}[!ht]
  \SetKwData{Left}{left}\SetKwData{This}{this}\SetKwData{Up}{up}
  \SetKwFunction{Union}{Union}\SetKwFunction{FindCompress}{FindCompress}
  \SetKwInOut{Input}{input}\SetKwInOut{Output}{output}

  \Input{$U(P,1)$ with minimum degree 2.}  \Output{$G$: Geometric Planar spanning subgraph of $U(P,2)$ of minimum degree 2 and longest edge length bounded by $2$.}

  \SetAlgoLined

  Let $T = (P,E)$ be the Nearest Neighbor Graph of $U(P,1)$.\\
  Let $E'$ be the set of SNN directed edges from leaves of $T$. \\
  Let $G = (P, E \cup E')$.\\

  \ForEach{edge $\{u,v\}$ in $G$ that forms a $Tie(u; v, u', v')$}
  { \If{$\{u,v\}$ does not form another $Tie$} {
      \lIf{$\triangle(uvu')$ is empty} {
        Let $E' = E' \cup \{ \{u, u'\} \} \setminus \{\{u', v'\}\}$.\\
      } \Else{
        Let $w \in \triangle(uvu')$ such that $\{u, w\} \in E$. \\
        Let $E' = E' \cup \{ \{w, u'\} \}\setminus \{\{u', v'\}\}$.  }
    } \If{$\{u,v\}$ forms a $Tie(v;, u, u'', v'')$ where $u''$ is a
      leaf of $T$} { \lIf{$\{u, u'\}$ crosses $\{u'', v\}$} {
        Let $E' = E' \cup \{ \{u, u' \},\{u'',v\}\} \setminus \{\{u',v'\},\{u'',v''\}\}$. \\
      } \lElse{ \lIf{the quadrangle $(uvu'u'')$ is empty} {
          Let $E' = E' \cup \{ \{u', u'' \}\} \setminus \{\{u',v'\},\{u'',v''\}\}$. \\
        } \Else{ Let $p$ and $q$ be the points in the quadrangle
          $(uvu'u'')$ such that $\angle(uu''p)$ and
          $\angle(qu'v)$ are minimum. \\
          Let $E' = E' \cup \{ \{u', p\}, \{q, u'\}  \} \setminus \{\{u',v'\},\{u'',v''\}\}$. \\
        } } } \If{$\{u,v\}$ forms a $Tie(u''; v'', u,v)$} { \If{the
        quadrangle $(uvu'u'')$ is empty} {
        Let $E' = E' \cup \{ \{u', u'' \}\} \setminus \{\{u',v'\}\}$. \\
        \lIf{$u''$ is a leaf of $T$} { Let $E' = E' \setminus
          \{\{u'',v''\}\}$. } } \Else{
        Let $p$ be the point in the quadrangle $(uvu'u'')$ such that $\angle(uu''p)$ is minimum. \\
        Let $E' = E' \cup \{ \{u'', p\}  \} \setminus \{\{u',v'\}\}$. \\
      }

    } } \ForEach{edge $\{u,v\}$ in $G$ that forms a $Bow(u, v, u',
    v')$} { Let $E' = E' \cup \{u, u'\} \setminus \{ \{u,v\},
    \{u',v'\}\}$ }

  \caption{Geometric planar subgraph of minimum degree 2 and
    longest edge length bounded by $2$.}
  \label{alg:alg1}
\end{algorithm}

The pseudocode is presented in Algorithm~\ref{alg:alg1}.  Regarding the complexity, the Nearest Neighbor Graph of $U(P,1)$ can be constructed in $O(n\log n)$. A range tree can be also constructed in $O(n\log n)$ where each query of proximity neighbors takes $O(\log n)$.  The removal of a crossing can be done in time $O(\log n)$ and there exist at most $2n$ $Tie$s since each leaf of $T$ can form at most two $Tie$s.  Therefore, the whole construction can be done in $O(n\log n)$ since there are at most $O(n)$ crossings. This complete the proof. \qed
\end{proof}

\section{2-Edge Connected Geometric Planar Subgraphs of a UDG of Minimum Degree 2}
\label{sec:2edge}

In this section we prove that if $U(P,1)$ is connected and has minimum degree 2, then $U(P,\sqrt{5})$ always contains a 2-edge connected planar spanning subgraph. We also show that the radius $\sqrt{5}$ is best possible. Therefore in this section we assume $U(P,1)$ is connected and has minimum degree 2.


The following theorem shows that the bound $\sqrt{5}$ is best possible. 

\begin{theorem}
\label{thm:thm2}
For any real $\epsilon>0$ and any integer $n\ge 8$, there exists a set $P$ of $n$ 
points in the plane so that $U(P,1)$ is connected and has minimum degree 2 but 
$U(P,\sqrt{5}-\epsilon)$ has no geometric planar $2$-edge connected spanning subgraph.
\end{theorem}

\begin{proof}
Consider the component $C$ despited in Figure~\ref{fig:cactuslower}. The vertex $x$ is called the entry point and has the following properties: $d(x)=1$, $d(v,x) \geq \sqrt{5}$ and $\{u_2,x\}$ crosses $C$.  Observe that $C$ requires at least one of the edges $\{u_1,w\},\{u_2,w\}$ be included so that the edge $\{v,w\}$ 
is in a 2-edge connected geometric planar spanning subgraph. We may assume without loss of generality that the edge $u_1w$ is added. Observe, that for any arbitrarily small $\epsilon > 0$, there exists $\delta>0$ sufficiently close to zero such that $\sqrt{5} - d(u_1,w) \leq \epsilon$. Observe that $C\setminus x$ has minimum degree two and the lower bound holds.  We can construct a family of UDGs with $n>8$ vertices and minimum degree two having the same lower bound by connecting the entry point $x$ to distinct UDG components.
  \begin{figure}[!ht]
    \begin{center}
      \includegraphics[scale=.7]{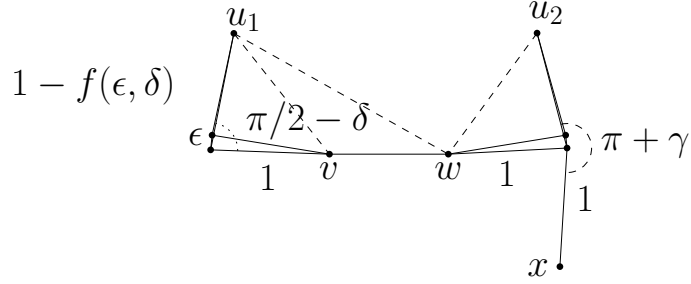}
    \end{center}
    \caption{UDG Component with minimum degree 2 that requires scaling
      factor of $\sqrt{5}$.}
    \label{fig:cactuslower}
  \end{figure}
  \qed
\end{proof}


\begin{theorem}\label{thm:main2}
Let $P$ be a set of $n$ points in the plane in general position such that $U(P, 1)$ is connected and has minimum degree 2. 
Then $U(P, \sqrt{5})$ has a 2-edge connected geometric planar spanning subgraph.  Further, it can be constructed in time $O(n\log n)$.
\end{theorem}

\begin{proof}
Let $T=(P,E)$ be a minimum spanning tree (MST) of $U(P,1)$. Properly color the internal vertices of $T$ with two colors, say black and red, and then color leaves with green. 
Recall that a proper $k$-coloring is an assignment  of one color among $k$ to vertices  in such a way that vertices of the same color are never adjacent.
Let $G=(P, E \cup E')$ be the spanning planar subgraph of $U(P,2)$ (which is a subgraph of $U(P,\sqrt{5})$) with minimum degree 2 obtained by Theorem~\ref{thm:main}. Choose a chromatic class, say black.  Consider a black vertex $u$ and its neighbor $v$ in $G$.  It is not difficult to see that if $\{u,v\} \in E'$, then $v$ is green, i.e. a leaf in $T$, and either $u$ was the tip of a $Tie(u, u', v, v')$ and $d(u,v) \leq \sqrt{2}$ or all the neighbors of $u$ in $T$ are in the same half-plane determined by $\{u,v\}$.
	 
Suppose that $\{u, v\} \in E$ is a bridge of $G$.  Consider the immediate edge $\{u, w\}$ of $\{u,v\}$ such that $\angle{wuv} < \pi$ with the preference to edges in $E$ and then edges in $E'$. We will add a new edge (for each such bridge) into $G$ and make sure these new edges do not add any crossings. The set of added edges will be $E''$ which is empty at the beginning.
 
\begin{itemize}
\item{$\{u, w\} \in E$.} 
Let $E'' = E'' \cup \{\{v,
    w\}\}$. Obviously $d(u,w) \leq 2$.

\item{$\{u, w\} \in E'$.} 
Observe that this corresponds to a  $Tie(u, u', w, w')$ as depicted in Figure~\ref{fig:2edge}.  We consider two cases: If $\triangle(uvw)$ is empty, then let $E'' =
    E'' \cup \{\{v,w\}\}$.  Otherwise, let $p$ and $q$ be the points such that $\angle(pvu)$ and $\angle(qwu)$ are minimum. Let $E'' =
    E'' \cup \{\{v,p\}, \{q,w\}\}$.  Since $u$ is the tip of a $Tie(u, u', v, v')$, from Lemma~\ref{lem:length}, $d(w, v) \leq
    \sqrt{5}$.
\end{itemize}
	
\begin{figure}[!ht]
  \centering
  \includegraphics[scale=.7]{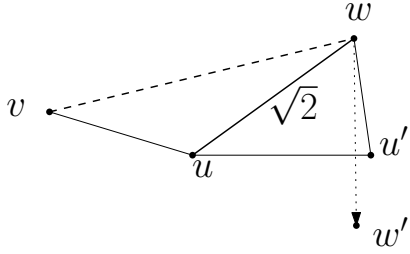}
  \caption{$\angle(wuv) < \pi$ and $\{u, v'\} \in E'$.}
  \label{fig:2edge}
\end{figure}

Observe that every vertex of $G =(P, E \cup E' \cup E'')$ is in at least one cycle. Therefore, it is two edge connected.  The pseudocode is presented in Algorithm~\ref{alg:alg2}.  Regarding to the complexity, each new edge can be added in time $O(\log{n})$. Therefore, the whole construction can be completed in time $O(n\log{n})$. \qed

\LinesNumbered
\begin{algorithm}[!ht]
  \SetKwData{Left}{left}\SetKwData{This}{this}\SetKwData{Up}{up}
  \SetKwFunction{Union}{Union}\SetKwFunction{FindCompress}{FindCompress}
  \SetKwInOut{Input}{input}\SetKwInOut{Output}{output}

  \Input{Connected UDG with minimum degree 2.}  \Output{$G$: 2-Edge
    Connected Planar Graph with longest edge length bounded by
    $\sqrt{5}$.}

  \SetAlgoLined

  Let $G = (P, E \cup E')$ be the connected planar graph of minimum
  degree 2
  obtained from Algorithm \mbox{\ref{alg:alg1}}. \\
  Color internal vertices of $T=(P,E)$ with black and red.\\

  \ForEach{Bridge $\{u, v\}\in E$ of $G$} {
    Let $u$ be a black vertex.\\
    Let $\{u, w\}$ be the immediate of $\{u, v\}$ such that
    $\angle{vwu} < \pi$
    with the preference to edges in $E$ and then edges in $E'$. \\
    \lIf{$\triangle(uvw)$ is empty} {
      Let $E' = E' \cup \{\{v, w\}\}$. \\
    } \Else{ Let $p$ and $q$ be the points in $\triangle(uvw)$ such
      that $\angle(uvp)$ and
      $\angle(qwu)$ are minimum. \\
      Let $E' = E' \cup \{ \{v, p\}, \{q,w\} \}$.  } }

  \caption{Constructing a 2-Edge Connected Planar Graph with longest
    edge length $\sqrt{5}$}
  \label{alg:alg2}
\end{algorithm}

\end{proof}

\section{2-Edge Connected Planar Subgraphs of a 2-Edge Connected UDG}
\label{sec:2edge2}

In this section we prove that if $U(P,1)$ is 2-edge connected, then $U(P,2)$ always contains a 2-edge connected geometric 
planar spanning subgraph. We also show that the radius $2$ is best possible. Therefore in this section we assume $U(P,1)$ is 2-edge connected.


The following theorem shows that the bound $2$ is best possible.

\begin{theorem}
\label{thm:thmlower3}
For any real $\epsilon>0$ and any integer $k$, there exists a set $R$ of $n=3k+1$ points in the plane so that $U(P,1)$ is 2-edge connected but $U(R,2-\epsilon)$ has no planar $2$-edge connected spanning subgraph.
\end{theorem}

\begin{proof}
The construction is based on the component depicted in Figure~\ref{fig:2edge1}. Observe that the component is the same as the component of the lower bound of planar graphs with minimum degree two. Clearly, it requires $\{u,v\}$ to create a 2-edge connected planar graph.  A UDG with $k$ components can be created by forming a convex path as depicted in Figure~\ref{fig:min}.  It is not difficult to see that the lower bound also holds for this UDG with $1+3k$ vertices.
  \begin{figure}[!ht]
    \begin{center}
      \subfloat[Basic component that requires scaling factor of
      $2$.]{\label{fig:2edge1}\includegraphics[scale=.7]{FIG/2-udg}}
      \hspace{0.5cm} \subfloat[Components forming a convex path with
      $1+3k$
      vertices.]{\label{fig:min}\includegraphics[scale=.5]{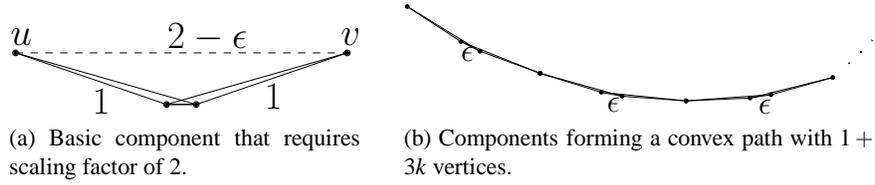}}
    \end{center}
    \caption{Two-edge connected UDG with $1+3k$ vertices that requires
      scaling factor of $2$.}
  \end{figure}
  \qed
\end{proof}


We say that a vertex $v$ of a graph $G$ is $Arduous$ if $v$ has degree two, is not in a cycle, and the angle that it forms with its consecutive neighbors is greater than $5\pi/6$. Thus, we have the following Corollary to Theorem~\ref{thm:thm2}.

\begin{corollary} 
\label{cor:1} 
Let $P$ be a set of $n$ points in the plane in general position such that $U(P, 1)$ is connected and has minimum degree 2.  Let $T = (P, E)$ be an MST of $U(P,1)$. Consider a (proper) 2-coloring of vertices of $T$ with colors black and red. If $U(P,1)$ does not have either black or red $Arduous$ vertices, then $U(P, 2)$ has an underlying 2-edge connected geometric planar graph.
\end{corollary}

\begin{proof}
Let $G=(P, E \cup E')$ be the 2-edge connected geometric planar spanning subgraph  obtained by Theorem \mbox{\ref{fig:dmin2}}.  Assume that $T$ does not have black $Arduous$ vertices.  For the sake of contradiction assume that $G$ has an edge $\{v, w\} \in E'$ such that $d(v, w) >2$. Let $u$ be the black vertex of $T$ that added $\{v, w\}$ to $G$.  Observe that $u$ was the tip of a $Tie(u; u', w, w'\}$ where $w$ is a leaf and the angle that $u$ forms with $u'$ and $w$ is greater than $5\pi/6$. However, $T$ does not have black $Arduous$ vertices.  This contradicts the assumption.
\qed
\end{proof}

First we prove that if $U(P,1)$ is 2-vertex connected, then $U(P, 2)$ has a spanning 2-edge connected geometric 
planar subgraph. Then we prove the same from 2-edge connectivity of $U(P,1)$. 

\begin{theorem}
\label{thm:main3}
Let $P$ be a set of $n$ points in the plane in general position such that $U(P, 1)$ is 2-vertex connected.  Then $U(P, 2)$ has a spanning 
geometric planar 2-edge connected subgraph.
\end{theorem}

\begin{proof}
Let $T =(P, E)$ be an MST of $U(P,1)$.  Consider a (proper) 2-coloring of internals vertices of $T$ with red and black colors, and assign green to leaves.  Choose any color class, say black. If $T$ does not have black $Arduous$ vertices, then by Corollary~\ref{cor:1}, $U(P, 2)$ has an underlying 2-edge connected planar graph.  Thus, assume that $T$ has at least one black $Arduous$ vertex. We will add edges to $E'$ in a greedy manner to obtain a graph $G=(P, E \cup E')$ that does not have black $Arduous$ vertices.

Consider a black $Arduous$ vertex $v$ of $G$. Let $\mathcal{G}_1$ and $\mathcal{G}_2$ be the connected components of $T\setminus v$ and $\{u,w\}$ be a shortest edge in $U(P, 1)$ that connects $\mathcal{G}_1$ and $\mathcal{G}_2$. Since $U(P, 1)$ is 2-vertex connected, $\{u, w\}$ always exists. Assume that $u \in \mathcal{G}_1$ and $w \in \mathcal{G}_2$.  Observe that every vertex in $D(u, d(u, w))$ is in $\mathcal{G}_1$ and every vertex in $D(w, d(u, w))$ is in $\mathcal{G}_2$, otherwise $\{u, w\}$ is not shortest.  Therefore, $D(u, d(u,w)) \cap D(w,d(u,w))$ either is empty or contains $v$.

We will show that $\{u, w\}$ does not cross an edge of $E$. For the sake of contradiction assume that $\{u, w\}$ crosses an edge $\{u', w'\} \in E$.  Let $R = D(u, d(u,w)) \cap D(w,d(u,w))$. Consider first the case when $u'$ and $w'$ are not in $R$. Therefore, either $\angle(u'uw)$ or $\angle(uwu')$ is the largest angle in $\triangle(uwu')$.  Similarly, either $\angle(wuw')$ or $\angle(w'wu)$ is the largest angle in $\triangle(uww')$. Observe that if $\angle(u'uw)$ and $\angle(wuw')$ are the largest angles, then there exists a cycle $u'w'u$ where $d(u', w')$ is the longest edge length. Therefore, $\{u',w'\}$ is not in $T$. Thus, assume that $\angle(u'uw)$ and $\angle(w'wu)$ are the largest angles in the respective triangles as depicted in Figure~\ref{fig:2ecg5}.  Hence, $d(u',w') > d(u,w)$. Therefore $d(u', u) \leq d(u,w)$ and similarly $d(w', w) \leq d(u,w)$.  This is a contradiction since there is a cycle $uww'u'u$ where $d(u', w')$ is the largest edge length. Now consider the case when at least one vertex of $u'$ or $w'$ is in $R$, say $w'$. Therefore, $v=w'$. However, $v$ is also incident to $u$ and $w$. This contradicts the assumption since $d(v) =2$.

Now we will prove that if $\{u,w\}$ crosses and edge $\{u',w'\} \in E'$, then $\{u',w'\}$ can be removed from $E'$ without increasing the number of black $Arduous$ vertices in $G$.  Assume without loss of generality that $u'$ and $w'$ are in $\mathcal{G}_1$ as depicted in Figure~\ref{fig:2ecg1}, otherwise, $v$ would not be an $Arduous$ vertex.  Therefore, $d(u,w) \leq \max(d(u', w), d(w, w'))$. Consider the previous step where $\{u', w'\}$ was added from $G'$. Let $v'$ be the black $Arduous$ vertex of $G'$ and $\mathcal{G'}_1$ and $\mathcal{G'}_2$ be the components of $G' \setminus v'$. Hence, $w$ was in either $\mathcal{G'}_{1}$ or $\mathcal{G'}_{2}$ and either $d(u',w') \leq d(u',w)$ or $d(u',w') \leq d(w',w)$. Therefore, they form a $Tie(w; u, u', w')$ where $u \in D(u'; d(u',w')) \cap D(w'; d(u', w'))$. Hence, $u=v'$. Thus, if $\{u,w\}$ crosses an edge $\{u',w'\} \in E'$, then let $E' = E' \cup \{\{u,w\}\} \setminus \{\{u', w'\}\}$.  Otherwise, let $E' = E' \cup \{\{u,w\}\}$.  Observe that any immediate neighbor $\{u, x\}$ and $\{w,y\}$ of $\{u, w\}$ where $x,y \notin D(u; d(u,w)) \cap D(w; d(u, w))$ form an angle of at least $\pi/3$.

\begin{figure}[!ht]
  \centering \subfloat[$\{u,w\}$ does not cross any edge of
  $T$.]{\label{fig:2ecg5}\includegraphics[scale=.7]{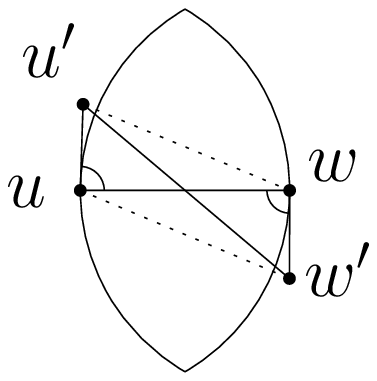}}
  \hspace{2cm} \subfloat[If $\{u,w\} \in E'$ crosses an edge
  $\{u',w'\} \in E'$, then $\{u', w'\}$ can be removed.]
  {\label{fig:2ecg1}\includegraphics[scale=.7]{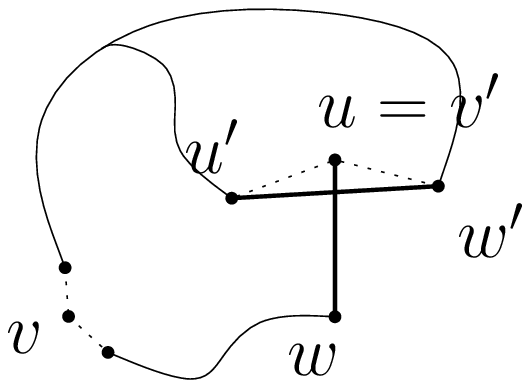}}
  \caption{Removal of black $Arduous$ vertices.}
\end{figure}

Clearly $G=(P, E \cup E')$ is planar and does not have black $Arduous$ vertices.  Let $E''$ be the set of SNN edges of $G$.

\begin{claim}
 Let $(u, v) \in E''$ be an edge that crosses an edge $\{u',v'\} \in E'$. \\
 (i) If $\{u, u'\},\{u,v'\} \notin E$, then $\{u',v'\}$ can be removed from $E'$ without increasing the number of black $Arduous$ vertices.\\
(ii) If $\{u,u'\},\{v,v'\} \in E$, then $\{u',v'\}$ can be removed from $E'$ without increasing the number of black $Arduous$ vertices. \\
(iii) If $\{u, u'\} \in E$ and $\{v, v'\} \notin E$, then they form a $Tie(v'; u', u,v)$.
\end{claim}

\begin{proof}[Claim]
Consider the step where $\{u', v'\}$ was added from $G'$. Let $w'$ be the black $Arduous$ vertex of $G'$ and let $\mathcal{G'}_1$ and $\mathcal{G'}_2$ be the components resulting from $G' \setminus w'$. Further, let $u' \in \mathcal{G'}_1$ and $v' \in \mathcal{G'}_2$.  Now we prove each case separately.

(i) Clearly $d(u, v) \leq \min(d(u, u'), d(u, v'))$ since $v$ is the second nearest neighbor of $u$. Assume without 
loss of generality that $u \in \mathcal{G'}_1$.  Therefore, $d(u', v') < d(u,v')$ and they form a $Tie(u;v,u',v')$. However, $v \in D(u'; d(u', v')) \cap D(v'; d(u',v'))$ which means that $w' = v$. Thus, we can remove $\{u',v'\}$ from $E'$ without increasing the number of black $Arduous$ vertices in $G$; see Figure \mbox{\ref{fig:2edgepr21}}.

(ii) First consider that $\{u',v\} \notin E$. Therefore, $d(u', v') < d(u',v)$ since $v$ is in the same component as $v'$. Observe that $\angle(uu'v')$ and $\angle(u'v'v)$ are the largest angles in the triangles $\triangle(uu'v')$ and $\triangle(u'v'v)$ respectively. However, since $d(u',v) \geq d(u', v')$ and $\angle(uu'v) > \angle(uu'v')$, $d(u, v') \leq d(u, v)$. This contradicts the assumption. Now consider that $\{u',v\} \in E$, then $vu'v'$ form a cycle where $\{u',v'\}$ is the longest edge otherwise $T$ is not minimum. Therefore, $v \in D(u'; d(u', v')) \cap D(v'; d(u',v'))$ and $w' = v$.  Thus, we can remove $\{u',v'\}$ from $E'$ without increasing the number of black $Arduous$ vertices in $G$.

(iii) First we will prove that $v \in \mathcal{G'}_1$.  Assume by contradiction 
that $v$ is in $\mathcal{G'}_2$.  Similarly to the previous case, $d(u', v') < d(u',v)$.  Thus, $\angle(uu'v')$ and $\angle(u'v'v)$ are the largest angles in the triangles $\triangle(uu'v')$ and $\triangle(u'v'v)$ respectively. 
However, since $d(u',v) > d(u', v')$ and $\angle(uu'v) > \angle(uu'v')$, $d(u, v') \leq d(u, v)$. Therefore, $u,v \in \mathcal{G'}_1$ and $d(u', v') \leq \min(d(v', u), d(v', v))$. Hence, they form a $Tie(v'; u', u, v)$ since $d(u, v') > d(u,v)$.

\begin{figure}[!ht]
  \centering \subfloat[If $\{u, u'\},\{u,v'\} \notin E$, then
  $\{u',v'\}$ can be removed from
  $E'$.]{\label{fig:2edgepr21}\includegraphics[scale=.6]{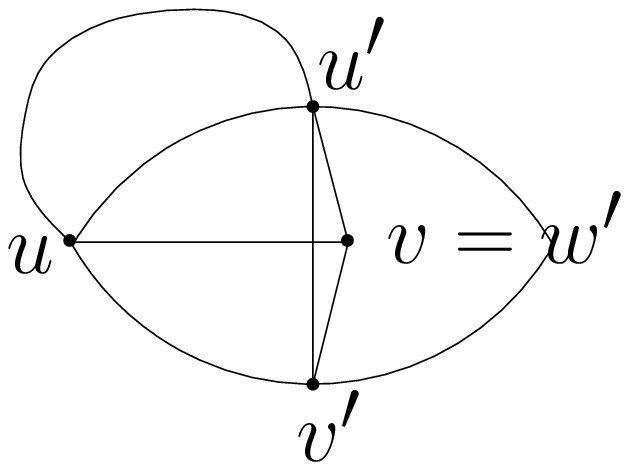}}
  \hspace{2cm}%
  \subfloat[If $\{u,u'\},\{v,v'\} \in E$, then $\{u',v'\}$ can be
  removed from
  $E'$.]{\label{fig:2edgepr22}\includegraphics[scale=.6]{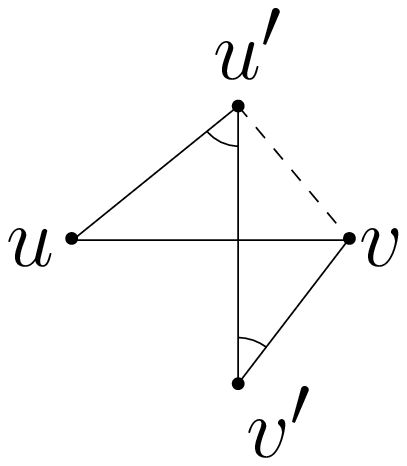}}
  \caption{Removal of black Arduous vertices. }
\end{figure}

\end{proof}

Observe that the crossings between edges in $E''$ and edges in $E \cup E'$ are equivalent to crossings between edges in $E''$ and $E$. 
That is, they form $Tie$s where leaves are endpoints of crossing lines. Thus, we can obtain a geometric planar graph of $G = (P, E \cup E' \cup E'')$ with minimum degree two from Theorem~\ref{thm:main}.  It remains to add each bridge of $G$ into at least one cycle.  Let $v$ be a black vertex of $G$ incident to a bridge $\{u, v\} \in E$ and $\{w,v\}$ be an edge such that $\angle(uvw) < \pi$ with the preference to edges in $E$, then in $E'$ and then in $E''$. We have three cases:
	
\begin{itemize}
\item{$\{w,v\} \in E$.} 
Let $E''= E'' \cup \{\{u,w\}\}$. Clearly, $d(u, w) \leq 2$.

\item{$\{w,v\} \in E'$.} 
We consider two cases. First assume that $w$ is red.  
Let $E''= E'' \cup \{\{u,w\}\}$. $d(u, w) \leq 2$. Now assume that $w$ is black. 
Clearly $d_G(v) \geq 3$ and $d_G(w) \geq 3$. Observe that since $\{w, v\} \in E'$ and $v$ is an internal 
black vertex of $T$, there exits a neighbor $w'$ of $v$ such that $\angle(uvw') < \pi$ and $\{u,w'\}$ crosses $\{v,w\}$.  
Therefore, $\angle(wvu) \leq 2\pi/3$.  Let $u'$ be the first neighbor of $w$ such that $u'wvu$ form a convex path; 
see Figure~\ref{fig:2edproof3}.  If either $u'$ does not exist or $\{u', w\} \in E'$ or $\{u', w\} \in E''$, then let $E'' = E'' \cup \{\{w, u\}\}$.  Otherwise, $\{u', w\} \in E$. Similarly, since $\{w, v\} \in E'$ and $w$ is an internal black vertex of $T$, 
there exits a neighbor $v'$ of $w$ such that $\angle(u'wv') < \pi$ and $\{u',v'\}$ crosses $\{w,v\}$.  
Therefore, $\angle(u'wv) \leq 2\pi/3$.  If the quadrangle $uvwu'$ is empty, then let $E'' = E'' \cup \{\{u,u'\}\}$.  Otherwise, let $p$ and $q$ be the points such that $\angle(pu'w)$ and $\angle(quv)$ are minimum. Let $E'' = E'' \cup \{\{u',p\}, \{q,u\}\}$.  It is not difficult to see that $d(u, u') \leq 2$. To see this, consider the right triangles $auv$ and $u'bw$ where $a$ and $b$ are the points in $\{u', u\}$ such that $\angle(vau) = \pi/2$ and $\angle(u'bw) = \pi/2$.  From the Law of sines $d(a, u) \leq 1/2$, $d(u', b) = 1/2$ and $d(p, q) =1$ since $\angle(avu) \leq \pi/6$ and $\angle(u'wb) \leq \pi/6$.

\begin{figure}[!ht]
  \centering
  \includegraphics[scale=.7]{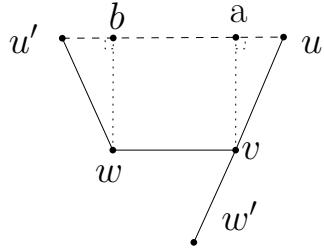}
  \caption{$\{w,v\} \in E'$ and $w$ is black.}
  \label{fig:2edproof3}
\end{figure}

\item{$\{w,v\} \in E''$.} 
We consider two cases: If $\triangle(uvw)$ is empty, then let $E'' = E'' \cup \{\{u,w\}\}$.  Otherwise, let $p$ and $q$ be the points such that $\angle(puv)$ and $\angle(qwv)$ are minimum. Let $E'' = E'' \cup \{\{u,p\}, \{q,w\}\}$.  Since $v$ is the tip of a $Tie(v, v', w, w')$ and $\angle(v'vu) \leq 5\pi/6$, from Lemma~\ref{lem:length}, $d(u, w) \leq 2$.
\end{itemize}

The pseudocode is presented in Algorithm~\ref{alg:alg3}.  Regarding the time complexity, the dominating step is the removal of $Arduous$ vertices and can be implemented in time $O(n^2)$.  That is, given an $Arduous$ vertex, determine the components $\mathcal{G}_1,\mathcal{G}_2$ of $G\setminus v$ in $O(n)$ time and look for the shortest edge length $\{u, w\}$ of $U(P, 1)$ not in $G$ such that $u \in\mathcal{G}_1$ and $w \in \mathcal{G}_2$ in $O(n)$ time.  Therefore, the construction can be done in $O(n^2)$ time. \qed

\LinesNumbered
\begin{algorithm}[!ht]
  \SetKwData{Left}{left}\SetKwData{This}{this}\SetKwData{Up}{up}
  \SetKwFunction{Union}{Union}\SetKwFunction{FindCompress}{FindCompress}
  \SetKwInOut{Input}{input}\SetKwInOut{Output}{output}

  \Input{2-vertex connected $U(P,1)$.}  \Output{$G$: Geometric planar 2-edge connected
    planar subgraph of $U(P,2)$with longest edge length bounded by $2$.}

  \SetAlgoLined

  Let $T=(P, E)$ be a MST of $U(P,1)$, $E' = \emptyset$ and  $G = (P, E \cup E')$. \\
  Color the internal vertices of $T$ with black and red colors. \\
  Let $A$ be the set of black $Arduous$ vertices of $T$. \\
  Let $G = (P, E \cup E')$. \\
  \While{$A$ is empty} {
    Let $v$ be a vertex of $A$ and $\mathcal{G}_1,\mathcal{G}_2$ be the components of $G \setminus v$. \\
    Let $\{u, w\}$ be the shortest edge such that $u \in \mathcal{G}_1$ and $w \in \mathcal{G}_2$ \\
    \lIf{$\{u, w\}$ crosses an edge $\{u', w'\} \in E'$} { Let $E' =
      E' \cup \{\{u, w\}\} \setminus \{\{u', w'\}\}$.  }  \lElse { Let
      $E' = E' \cup \{\{u, w\}\}$.  } Remove the vertices of $A$ that
    are in cycles or have degree at least three in $G$.  } Let $E''$
  be the SNN edges of $G$ and $G = (P, E \cup E' \cup E'')$ be the
  connected geometric planar graph of minimum degree 2
  obtained from Algorithm \mbox{\ref{alg:alg1}}. \\

  \ForEach{Black vertex $u \in T$} {
    Let $v$ be a black vertex and $\{v,u\}$ be a bridge of $G$.\\
    Let $\{v,w\}$ be the consecutive edge such that $\angle(wvu) <
    \pi$
    and given the following priority $E$, $E'$, $E''$.\\
    \lIf{$\{v,w\} \in E$} {
      Let $E' =  E'' \cup \{\{u, w\}\}$. \\
    } \lIf{$\{v,w\} \in E'$} { \lIf{$w$ is red} {
        Let $E' =  E'' \cup \{\{u, w\}\}$. \\
      } \If{$w$ is black} {
        Let $u'$ be the first neighbor of $w$ such that $u'wvu$ form a convex path.\\
        \lIf{$u'$ does not exist or $\{w, u'\} \in E'$ or $\{w, u'\}
          \in E''$} {
          Let $E' =  E'' \cup \{\{u, w\}\}$. \\
        } \Else{ \lIf{The quadrangle $u'wvu$ is empty}{
            Let $E' =  E'' \cup \{\{u, u'\}\}$. \\
          } \Else{ Let $p$ and $q$ be the points in $u'wvu$ such that
            $\angle(pu'w)$ and
            $\angle(quv)$ are minimum;\\
            Let $E' =  E'' \cup \{\{u', p\}, \{q, u\}\}$. \\
          } } }

    } \lIf{$\{v,w\} \in E''$} { \lIf{$\triangle(uvw)$ is empty} {
        Let $E' = E' \cup \{\{u, w\}\}$. \\
      } \Else{ Let $p$ and $q$ be the points in $\triangle(uvw)$ such
        that $\angle(vwp)$ and
        $\angle(quv)$ are minimum. \\
        Let $E' = E' \cup \{ \{w, p\}, \{q, u\} \}$.  } } }

  \caption{Geometric planar 2-Edge connected subgraph with longest
    edge length bounded by $2$}
  \label{alg:alg3}
\end{algorithm}

\end{proof}

\begin{theorem}\label{thm:main4}
  Let $P$ be a set of $n$ points in the plane in general position such that $U(P, 1)$ is 2-edge connected.  Then $U(P, 2)$ has a spanning geometric
planar 2-edge connected subgraph.
\end{theorem}

\begin{proof}
Consider the subsets $P_i$ of $P$ such that $U(P_i, 1)$ is 2-vertex connected.  Using Theorem~\ref{thm:main3}, we can construct a spanning 2-edge 
connected geometric planar subgraph $G_i$ of $U(P_i, 2)$ since each $U(P_i, 2)$ has at least three vertices. 
 It is not difficult to see that $\bigcup G_i$ is 2-edge connected and planar. \qed
\end{proof}

\section{UDG of High Connectivity without 2-Edge Connected Geometric Planar Subgraphs}
\label{sec:highconnected}

One may ask: for which $k>1$, a $k$-edge (or $k$-vertex) connected $U(P,1)$ with $n$ points has a 
spanning 2-edge connected geometric planar subgraph? We will show that even for $k \in O(\sqrt{n})$ this is not always true.

\begin{theorem}\label{thm:highcon}
 There exist a set $P$ of $n$ points in the plane so that $U(P,1)$ is $k$-vertex connected,  $k \in O(\sqrt{n})$,  
but $U(P,17/16)$ does not contain any  2-edge connected geometric planar spanning subgraph.
\end{theorem}

\begin{proof}
Assume $k=2m$.  Consider the $C^k$ and the wire components depicted in Figure~\ref{fig:pair} and Figure~\ref{fig:2edgwire} with $2k+2$ vertices and $2k$ vertices respectively.  It is easy to see that $C^k$ is a valid two-vertex connected UDGs and the wire is a valid $k$-vertex connected UDGs.  Observe that $C^k$ does not have a 2-edge connected planar subgraph since the inclusion of $\{u_1, u'_k\}$ and $\{u'_1, u_k\}$ leaves $v'$ and $v$ with degree one respectively. Hence, we call $v$ and $v'$ the isolated vertices of $C^k$.  Observe that we can embed $C^k$ in such a way that the distances $d(v, u_k),d(u_k, u_{k-1}),d(u_2, u_1)$ and $d(v', u'_k),d(u'_k, u'_{k-1}),d(u'_2, u'_1)$ are $\frac{1}{4}-\epsilon$. Hence, $d(u'_1, v)=d(u'_1, u_2) = d(u_1, v')=d(u_1, u'_2) = 17/16 - \delta$.  
Let $C^k_i$ be $m$ consecutive $C^k$ components in such a way that they are at distance greater than $17/16$ from each other.  We can connect the upper and lower part of $C^k_i$ with $C^k_{i+1}$ with a constant number of wires, i.e. creating $k$ independent paths that connect the upper and the lower part of $C^k_i$ and $C^k_{i+1}$ in such a way that the isolated vertices of each $C^k_i$ are far from the wires as depicted in Figure~\ref{fig:pair}.  It is easy to see that the resulting graph is $k$-vertex connected and has $O(k^2)$ vertices. 

  \begin{figure}[!ht]
    \begin{center}
      \subfloat[$C^k$
      component.]{\label{fig:pair}\includegraphics[scale=.6]{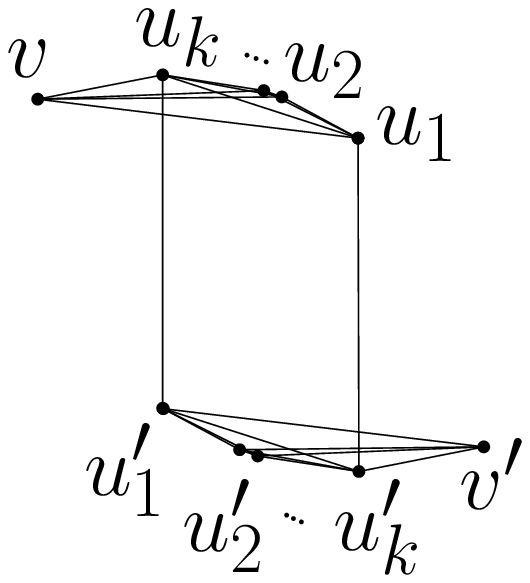}}
      \hspace{1cm}%
      \subfloat[Wire.]{\label{fig:2edgwire}\includegraphics[scale=.6]{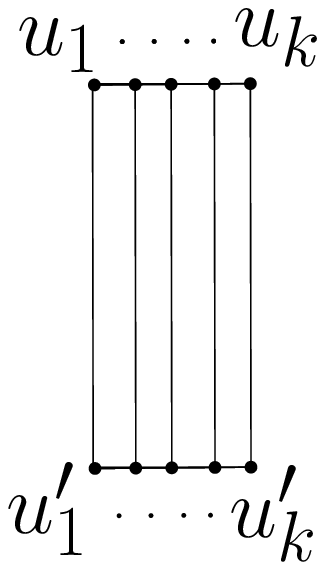}}
      \hspace{1cm}%
      \subfloat[Upper connection between $C^k_i$ and
      $C^k_{i+1}$.]{\label{fig:pair2}\includegraphics[scale=.7]{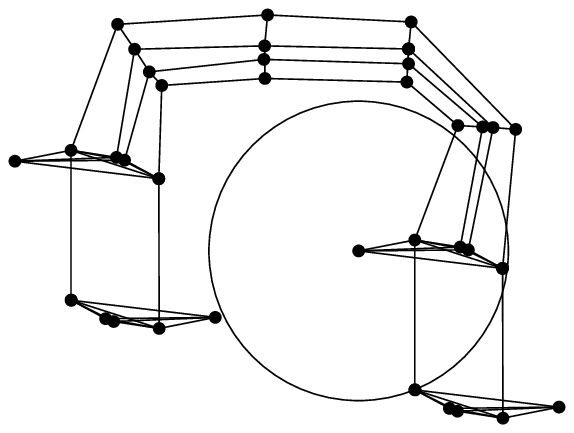}}
    \end{center}
    \caption{$k$-vertex connected UDG that does not have 2-edge
      connected planar subgraph.}
  \end{figure}
\qed
\end{proof}

\section{Conclusion}

In this paper, we have shown that for any given point set 
$P$ in the plane forming a 2-edge connected unit disk graph, 
the geometric graph $U(P,2)$ contains a 2-edge connected 
geometric planar graph that spans $P$. It is an open problem 
to determine necessary and sufficient conditions for
constructing $k$-vertex (or $k$-edge) connected planar
straight line edge 
graphs with bounded edge length on a set of points 
for $3 \leq k \leq 4$.

\bibliographystyle{plain}
\bibliography{biblio}

\end{document}